%% file: main.tex
\documentclass[conference,compsoc]{IEEEtran}
\ifCLASSOPTIONcompsoc
  \usepackage[nocompress]{cite}
\else
  \usepackage{cite}
\fi
\ifCLASSINFOpdf
\else
\fi
\usepackage{amsmath,amssymb}
\usepackage{mathtools}
\usepackage{tikz}
\usetikzlibrary{arrows,fit,positioning}
\usepackage{pgfplots}
\pgfplotsset{width=10cm,compat=1.9}
\usepackage[caption=false,font=footnotesize]{subfig}
\usepackage[export]{adjustbox}
\usepackage{tabularx}
\usepackage{colortbl}
\usepackage{multirow}
\usepackage{stmaryrd}
\usepackage{arydshln}
\usepackage{amsmath,amssymb}
\usepackage{mathtools}
\usepackage{formal-grammar}
\usepackage{tikz}
\usetikzlibrary{arrows,fit,positioning,shapes,patterns}
\usepackage{pgfplots}
\usepackage{forest}
\usepackage{soul}
\usepackage{listings}
\usepackage{flushend}
\pgfplotsset{width=8cm,compat=1.9}
\definecolor{scgreen}{HTML}{7fc97f}
\definecolor{scorange}{HTML}{fdc086}
\definecolor{scpurple}{HTML}{beaed4}
\definecolor{scred}{HTML}{e41a1c}
\definecolor{scdgreen}{HTML}{4daf4a}
\usepackage{algorithm}
\usepackage{algpseudocode}

\usepackage{soul}

\usepackage{array}
\usepackage{tabularx}
\usepackage{colortbl}
\usepackage{multirow}
\usepackage{arydshln}
\usetikzlibrary{matrix}

\usepackage[caption=false,font=footnotesize]{subfig}
\usepackage[export]{adjustbox}
\usepackage{soul}

\newcommand{\benc}{B}
\newcommand{\maxareaboothbase}{31\%}
\newcommand{\maxpowerboothbase}{38\%}
\newcommand{\maxareaboothand}{13\%}
\newcommand{\maxpowerboothand}{12\%}
\newcommand*\best[1]{{\bf #1}}
\usepackage{amsthm}
\newtheorem{theorem}{Theorem}
\newtheorem{lemma}[theorem]{Lemma}

\begin{document}
%

\title{On the Systematic Creation of Faithfully Rounded \\ Commutative Truncated Booth Multipliers}

\author{
\IEEEauthorblockN{Theo Drane, Samuel Coward}
\IEEEauthorblockA{Graphics Numerical Hardware Group\\
Intel Corporation\\
Email: \{theo.drane,samuel.coward\}@intel.com}
\and
\IEEEauthorblockN{Mertcan Temel}
\IEEEauthorblockA{Advanced Architecture Design Group\\
Intel Corporation\\
Email: mert.temel@intel.com}
\and
\IEEEauthorblockN{Joe Leslie-Hurd}
\IEEEauthorblockA{Design Engineering Group\\
Intel Corporation\\
Email: joe.leslie-hurd@intel.com}
}



\maketitle

\begin{abstract}
In many instances of fixed-point multiplication, a full precision result is not required. Instead it is sufficient to return a faithfully rounded result. Faithful rounding permits the machine representable number either immediately above or below the full precision result, if the latter is not exactly representable. Multipliers which take full advantage of this freedom can be implemented using less circuit area and consuming less power. The most common implementations internally truncate the partial product array. However, truncation applied to the most common of multiplier architectures, namely Booth architectures, results in non-commutative implementations. The industrial adoption of truncated multipliers is limited by the absence of formal verification of such implementations, since exhaustive simulation is typically infeasible. We present a commutative truncated Booth multiplier architecture and derive closed form necessary and sufficient conditions for faithful rounding. We also provide the bit-vectors giving rise to the worst-case error. We present a formal verification methodology based on ACL2 which scales up to 42 bit multipliers. We synthesize a range of commutative faithfully rounded multipliers and show that truncated booth implementations are up to \maxareaboothbase~smaller than externally truncated multipliers.
\end{abstract}


\IEEEpeerreviewmaketitle

\section{Introduction}\label{sec:intro}
Of the most common arithmetic circuits, multiplication consumes the greatest power and occupies the largest circuit area. As a result, binary multiplication has been the subject of significant academic and industrial research~\cite{Wallace1964AMultiplier,Dadda1965SomeMultipliers,Zimmermann2003OptimizedSum-of-products}. Amongst the most widely implemented multiplier architectures is the Booth Radix-4 multiplier~\cite{Ercegovac2004DigitalArithmetic}. In many applications, the requirement for exact multiplication can be dropped, and replaced with a faithful rounding requirement. A faithful rounding returns the machine representable number immediately above or below the infinitely precise result, unless the infinitely precise result can in fact be represented at the machine precision.

The additional freedom introduced by a faithful rounding, can be exploited, at the register transfer level (RTL), to improve multiplier power consumption and save circuit area. The standard approach to exploiting such freedom is to truncate the partial product array, as shown in Figure~\ref{fig:trunc_and_pic}. Unfortunately, the partial product array arising from a Booth Radix-4 encoding is not symmetric, as only one operand gets encoded. Applying truncation to a Booth Radix-4 multiplier, results in a non-commutative implementation~\cite{Huang2006AApplications,Chen2010AApplications}. Compiler optimizations routinely implicitly assume mathematical properties of underlying hardware. Application level correctness may implicitly require monotonicity of a complex function say, but may not have been considered during the hardware design. But commutativity is a far greater pervasive assumption, to the extent that compiler engineers could not conceive that non-commutative multipliers could even be built. Preserving commutativity significantly reduces compiler complexity for this most fundamental of operations.

\begin{figure}
    \centering
    \input{trunc_pp_array}
    \caption{A 16-bit $a\times b$, implemented as a traditional partial product array, where each partial product bit is $a[i]\& b[j]$. From this array we truncate 12 columns and insert the compensation constant, 11 (red). Summing the truncated array (black) and discarding the light blue bits produces a faithfully rounded 16-bit multiplication result.}
    \label{fig:trunc_and_pic}
\end{figure}
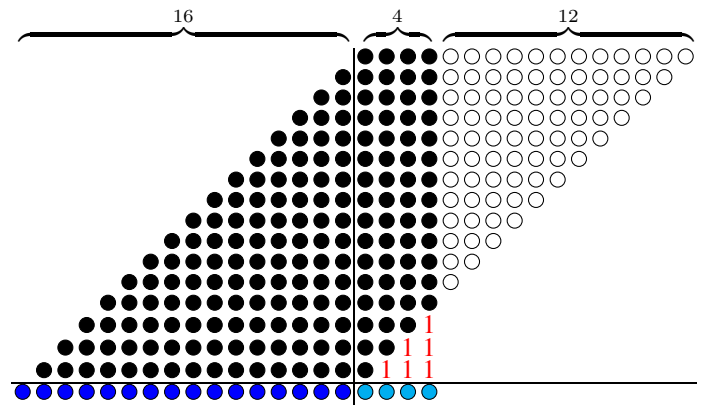
In this work, we first demonstrate how, for minimal hardware overhead, the commutativity property of a truncated Booth implementation can be recovered. We then analytically derive tight error bounds on Booth array truncation and describe necessary and sufficient conditions to implement a faithful rounding. Lastly, we describe a procedure to construct efficient hardware implementations. The result is a fully parameterizable faithfully rounded multiplier design that exhibits better power and area than alternative approaches.

The paper is organized as follows. In Section~\ref{sec:background} we discuss prior work on hardware efficient multiplier, with a particular focus on works that exploit error freedom. In Section~\ref{sec:method} we derive compensation terms to recover commutativity and prove a set of mathematical bounds on truncation error. We then demonstrate how we can use this information to design an efficient faithfully rounded commutative Booth multiplier. In Section~\ref{sec:results} we compare synthesis results for a range of different implementations and discuss our approach to verifying these multiplier designs.

The paper contains the following novel contributions:
\begin{itemize}
    \item a precise description of the compensation hardware required to recover the commutativity property,
    \item proven mathematical bounds on the maximal error due to truncation of Booth partial product arrays,
    \item a procedure to construct faithfully rounded commutative truncated Booth multipliers with maximal truncation.
\end{itemize}

\section{Background}\label{sec:background}
\subsection{Binary Multiplication}
The most naive implementation of $n$-bit binary multiplication via primitive logic gates~\cite{Ercegovac2004DigitalArithmetic} consists of first forming a partial product (PP) array of $n^2$ PP bits, where each PP bit is the logical AND of two input bits. Next that PP array is reduced from $n$ rows to two rows via compressor cells arranged in a reduction tree. Once reduced to a summation of two rows, a full parallel adder can be deployed, usually to produce a $2n$-bit result. Decades of research has led to efficient implementations of all stages of this approach. Array reduction was studied by Wallace and Dadda~\cite{Wallace1964AMultiplier,Dadda1965SomeMultipliers} then more recently improved with timing driven compressor-tree construction~\cite{Oklobdzija1996AApproach,Zimmermann2009DatapathDesign}. Parallel addition in ASIC design is now most commonly implemented via parallel prefix structures~\cite{Zimmermann2009DatapathDesign,Zimmermann2003OptimizedSum-of-products}. In this work, we will be entirely focused upon the first step, the construction of the PP array, and rely on existing techniques for the efficient implementation of array reduction and parallel addition.

One of the most widely used PP array creation techniques, is to add a Booth encoding step, which groups bits together to reduce the number of rows in the PP array. As observed by Zimmerman~\cite{Zimmermann2009DatapathDesign}, the overhead introduced by the additional encoding step only offers a net benefit at larger bitwidths, beyond 16-bits. The two primary variants are the Booth Radix-4 and Radix-8 encoding schemes. This paper modifies the Booth Radix-4 multiplication method, so we will describe it in detail. We define the following encoding function taking three single bit operands:
\begin{equation}
    \benc(x,y,z) = -2x+y+z
\end{equation}
For an $n$-bit signed multiplication, Booth Radix-4 encoding halves the PP array height:
\begin{align}
    a\times b &= \sum_{i=0}^{n-1} 2^i a_i\times b \\
    &= \sum_{i=0}^{(n/2)-1} 4^{i} \benc(a_{2i+1},a_{2i},a_{2i-1})\times b, \label{eqn:booth_basic}
\end{align}
where $a_{-1}=0$ and we assume an even $n$. Each PP row in \eqref{eqn:booth_basic} can be efficiently implemented using primitive logic gates. Minor modifications are required for odd $n$ and unsigned multiplication~\cite{Ercegovac2004DigitalArithmetic}.

\subsection{Faithfully Rounded Binary Multiplication}
The implementation of binary multiplication where the full result is not required is most commonly achieved via truncation schemes~\cite{Drane2014OnArrays,Cho2004ErrorMultiplier,Ko2011DesignRounding}. These schemes follow a similar structure to compute $n$-bit $a\times b$. First, truncate the partial product array, removing the $k$ least significant columns, corresponding to an error value of $\Delta_k$. To the remaining array add a compensation term, $f(a,b)$, to the $k^{th}$ column and then perform standard array reduction. From this summation result a further $n-k$ columns can be truncated to recover a faithfully rounded multiplication. Early work in this domain, described as Constant Correction Truncated schemes (CCT), started from an AND array and considered constant $f(a,b)$~\cite{Schulte1993TruncatedDSP,Kidambi1996Area-efficientApplications}. An example of CCT is shown in Figure~\ref{fig:trunc_and_pic}. Later work introduced Variable Correction Truncation, where $f(a,b)$ was considered to be a function of the inputs~\cite{King1998Data-dependentMultipliers}. Other works considered linearization schemes~\cite{Petra2011DesignofFunction} and approximate carry predictions~\cite{Michard2006CarryMultiplication}. Booth array truncation has similarly seen CCT techniques applied along with a range of statistical methods to approximate the expected truncation error~\cite{Huang2006AApplications,Chen2010AApplications,Drane2014OnArrays}. These previous Booth truncation schemes have broken the commutativity property of multiplication, a property that, as we show in this work, can be recovered for minimal hardware overhead. 

Faithfully rounded truncated multipliers are most commonly applied in Digital Signal Processing (DSP) but also can be found in floating point multipliers, where a lower accuracy can be permitted~\cite{Wires2000Variable-correctionMultipliers}. For transcendental function approximations, it is rarely necessary to compute full precision multiplication, since we already have the approximation error to factor in~\cite{Orloski2023AutomaticArchitectures,DeCaro2017MinimizingRounding}.
\section{Methodology} \label{sec:method}
The cause of the non-commutativity of truncated Booth multiplier architectures stems from the asymmetry of how the multiplier and multiplicand are treated; namely that only one of the inputs is Booth encoded. Naturally to maintain the hardware benefits of Booth architectures and remain commutative, the solution is to Booth encode both inputs. For ease of exposition we will assume that both inputs are signed two's complement, have even bitwidth $n$. We focus on a Booth Radix-4 architecture. It is simple to extend the analysis presented here to odd $n$ and unsigned multiplication. Given these assumptions, \textit{double} Booth encoding results in the following:
\begin{align} \label{doublebooth}
a\times b &= \sum_{i,j=0}^{(n/2)-1} 4^{i+j} PP_{i,j}\\
PP_{i,j} & = \benc(a_{2i+1},a_{2i},a_{2i-1}) \times \benc(b_{2j+1},b_{2j},b_{2j-1}),
\end{align}
where $a_{-1}=b_{-1}=0$. The key observation to achieving commutative truncated Booth architectures is to truncate \eqref{doublebooth} directly. Ultimately we will target creating a faithfully rounded result returning the most significant $n$ bits. By truncating $k$ columns, a portion of the summation, $\Delta$, will be deleted and the remainder, $M$, will be implemented as a partial-product array such that $a \times b = M + \Delta$.
\begin{align} \label{doubleboothtrunc}
M &= \sum_{i,j=0, i+j\geq k/2}^{(n/2)-1} 4^{i+j} PP_{i,j} \nonumber \\ 
\Delta &= \sum_{i,j=0, i+j < k/2}^{(k/2)-1} 4^{i+j} PP_{i,j} 
\end{align}
Note that the proposed architecture and analysis assumes $k<n$ such that $\Delta$ is strictly triangular in shape. (Hence the proposed architecture can be extended to return a faithfully rounded result with $m$ bits as long as $m>n$.) The first challenge is how $M$ is transformed into a binary array and the second is analyzing the error of $\Delta$.

\subsection{Commutative Truncated Booth Arrays}\label{subsec:compensation}

To begin reducing $M$ to a binary PP array, we can undo the Booth encoding to one of the inputs:
\begin{align} \label{reduceM}
M &= \sum_{i=0}^{(n/2)-1} 4^i \benc(a_{2i+1},a_{2i},a_{2i-1}) \times bb_i \\
bb_i&=(-2^{n-k+2i-1}b_{n-1}+b_{n-2:k-2i}+b_{k-2i-1}),\nonumber
\end{align}
where $b_j=0$ if $j<0$ and $b_{n:m}$ denotes the bit slice of $b[n:m]$. Now the terms in the summation only differ from a standard Booth Radix-4 summation due to the presence of the $b_{k-2i-1}$ term. Let us consider a standard Booth Radix-4 PP row $row_i$ and contrast this with one of the terms in \eqref{reduceM}, $row_i'$. Such rows are of the form (for some $c_i$):

\begin{align} \label{pps}
row_i  &=  \benc(a_{2i+1},a_{2i},a_{2i-1}) (-2^{n-2}c_{n-1}+c_{n-2:1}) \nonumber \\ 
row_i' &=  \benc(a_{2i+1},a_{2i},a_{2i-1}) (-2^{n-2}c_{n-1}+c_{n-2:1}+c_0) 
\end{align}
Now the bit level construction of $row_i$ and $row_i'$ can be found in Table~\ref{tab:pp_creation} (note that for two's complement $c$ and bit $c_0$, $-c-c_0=\overline{c}+1-c_0=\overline{c}+\overline{c_0}$).
\begin{table*}
\begin{center}
\caption{PARTIAL PRODUCT CREATION FOR COMMUTATIVE TRUNCATED BOOTH RADIX-4 ARRAY}\label{tab:pp_creation}
\begin{tabular}{|c|c|c|c|c|c|}
\hline
$a_{2i+1}$ & $a_{2i}$ & $a_{2i-1}$ & $-2a_{2i+1}+a_{2i}+a_{2i-1}$  & $row_i=pp_i+s_i$ & $row_i'=pp_i'+s_i'$ \\
\hline
0 & 0 & 0 & 0 & 0 & 0 \\
0 & 0 & 1 & 1 & $\{c_{n-1},c_{n-1} ... c_2,c_1\}+0$ & $\{c_{n-1},c_{n-1} ... c_2,c_1\}+c_0$ \\
0 & 1 & 0 & 1 & $\{c_{n-1},c_{n-1} ... c_2,c_1\}+0$ & $\{c_{n-1},c_{n-1} ... c_2,c_1\}+c_0$ \\
0 & 1 & 1 & 2 & $\{c_{n-1},c_{n-2} ... c_1,0\}+0$ & $\{c_{n-1},c_{n-2} ... c_1,c_0\}+c_0$ \\
1 & 0 & 0 & -2 & $\{\overline{c_{n-1}},\overline{c_{n-2}} ... \overline{c_1},1\}+1$ & $\{\overline{b_{n-1}},\overline{c_{n-2}} ... \overline{c_1},\overline{c_0}\}+\overline{c_0}$ \\
1 & 0 & 1 & -1 & $\{\overline{c_{n-1}},\overline{c_{n-1}} ... \overline{c_0},\overline{c_0}\}+1$ & $\{\overline{b_{n-1}},\overline{c_{n-1}} ... \overline{c_2},\overline{c_1}\}+\overline{c_0}$ \\
1 & 1 & 0 & -1 & $\{\overline{c_{n-1}},\overline{c_{n-1}} ... \overline{c_0},\overline{c_0}\}+1$ & $\{\overline{b_{n-1}},\overline{c_{n-1}} ... \overline{c_2},\overline{c_1}\}+\overline{c_0}$ \\
1 & 1 & 1 & 0 & 0 & 0 \\
\hline
\end{tabular}
\end{center}
\end{table*}

Note that $row_i$ and $row_i'$ can both be expressed in the form integer $pp$ plus bit $s$. Moreover truncated $pp_i$ has identical values to truncated $pp_i'$, the key and only difference between a truncated Booth Radix-4 array and commutative array are the $s_i'$ bits. The Boolean expressions for $s_i$ and $s_i'$, derivable from Table~\ref{tab:pp_creation} are:
\begin{align} \label{sbits}
s_i &=  a_{2i+1} \& \overline{a_{2i} \& a_{2i-1}} \nonumber \\ s_i' &= a_{2i+1} \& \overline{a_{2i} \& a_{2i-1}} \& (c_0 \oplus a_{2i+1}) \end{align}
We can now present the steps in constructing a commutative truncated Booth Radix-4 array:
\begin{enumerate}
    \item Construct standard Booth Radix-4 array
    \item Remove least significant $k$ columns, $k$ is even
    \item  In column $k$, for $i \in [0,\frac{k}{2}-1]$, add additional bits 
\begin{equation*}
    s_i' = a_{2i+1} \& \overline{a_{2i} \& a_{2i-1}} \& (b_{k-2i-1} \oplus a_{2i+1}).
\end{equation*}
\end{enumerate}
The difference between a commutative and a non-commutative array is just the inclusion of these $k/2$ bits.
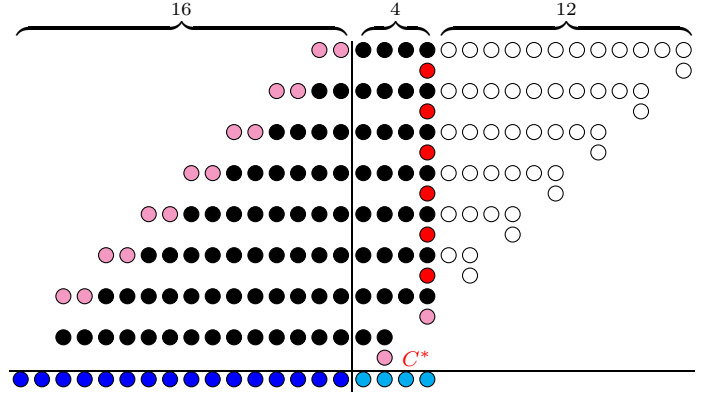
\begin{figure}
    \centering
    \input{compensated_booth}
    \caption{A 16-bit commutative truncated Booth multiplier, with 12 columns of truncation. The six red bits are the additional compensation bits $s_i'$. The pink bits represent the typical Booth sign bits $s_i$.} 
    \label{fig:compensated_booth}
\end{figure}
Since these $s'$ bits increase the array height of the least significant column, it is natural to ask whether this will affect the critical path.

Consider any array reduction that performs the summation of $m$ addends of large length $k$. The greatest number of carries that the summation generates will occur when the addends are maximal, evaluating $m(2^k-1) = (m-1)2^k+(2^k-m)$. For sufficiently large $k$, $2^k-m>0$ and hence such a summation will generate at most $m-1$ carries. This means the implementations are capable of dealing with $m-1$ carries and hence accepting $m-1$ additional carry-in bits in its least significant column without significantly altering its delay characteristics.

The maximum array height in Figure~\ref{fig:compensated_booth} is $\approx n/2$, the number of additional carry ins such an array can handle is $<n/2$, there are $k/2$ additional $s_i'$ bits which is by assumption on $k$, $<n/2$. Hence such commutative truncated Booth Radix-4 arrays are expected to exhibit minimal delay differences when compared to untruncated arrays.

It is important to note that while there are but $k/2$ $s'$ bits additional bits required to make the truncated Booth Radix-4 array commutative and such inclusion is expected to have limited delay impact; these bits are in \textit{no way} trivial. Omitting any one of them is likely to still produce a faithfully rounded implementation, but it will not be commutative. These $s'$ bits are not present in the original Booth Radix-4 array and any truncation or promotion of bits within the array will \textit{not} result in a commutative array. Such compensation bits would need to be rederived for every Booth architecture variant. Moreover commutativity was achieved by truncating \eqref{doublebooth}, not standard Booth summation formulae.

\subsection{The Truncation Error}

In order to create optimal faithfully rounded multipliers, the value range of $\Delta$ must be precisely known. This analysis will be facilitated by the following helper functions, representing hexadecimal summations:

\begin{center}
\scalebox{0.875}{
\begin{tabular}{rrrr}
\multicolumn{1}{l}{$X_n=$} & \multicolumn{1}{l}{$Y_n=$} & \multicolumn{1}{l}{$Z_n=$} & \multicolumn{1}{l}{$W_n=$}\\
$\overbrace{\texttt{222..222}}^n$ & $\overbrace{\texttt{222..222}}^n$ & $\overbrace{\texttt{000..000}}^n$ & $\overbrace{\texttt{444..444}}^n$ \\
                    & \texttt{+444..444}& \texttt{+444..444} & \texttt{+444..444} \\
\texttt{+444..440} & \texttt{+444..440} & \texttt{+444..440} & \texttt{+444..440} \\
\texttt{+444..440} & \texttt{+444..440} & \texttt{+444..440} & \texttt{+444..440} \\
\texttt{+444..400} & \texttt{+444..400} & \texttt{+444..400} & \texttt{+444..400} \\
\texttt{+444..400} & \texttt{+444..400} & \texttt{+444..400} & \texttt{+444..400} \\
           $\cdots$ & $\cdots$            & $\cdots$            & $\cdots$            \\
\texttt{+440..000} & \texttt{+440..000} & \texttt{+440..000} & \texttt{+440..000} \\
\texttt{+440..000} & \texttt{+440..000} & \texttt{+440..000} & \texttt{+440..000} \\
\texttt{+400..000} & \texttt{+400..000} & \texttt{+400..000} & \texttt{+400..000} \\
\texttt{+400..000} & \texttt{+400..000} & \texttt{+400..000} & \texttt{+400..000} \\
\end{tabular}
}
\end{center}
These hexadecimal summations can be simplified. Consider reducing $X_n$:
\begin{center}
\setlength\extrarowheight{4pt}
\scalebox{0.875}{
\begin{tabular}{rrr}
$\overbrace{\texttt{222..222}}^n$ & $=\frac{2}{15} \texttt{FFF..FFF}$ & $=\hspace{1.5em}\frac{2}{15} \overbrace{\texttt{0FFF..FFF}}^{n+1}$ \\
\texttt{+444..440} & $-\frac{8}{15} \texttt{111..110}$ & $-\hspace{1.5em}\frac{8}{15^2} \texttt{0FFF..FF0}$ \\
\texttt{+444..440} & $+\frac{8}{15} \texttt{111..110}$ & $+\frac{8(n-1)}{15^2} \texttt{1000..000}$ \\
\texttt{+444..400} & $+\frac{8}{15} \texttt{FFF..FF0}$ &  \\
\texttt{+444..400} & $+\frac{8}{15} \texttt{FFF..F00}$ &  \\
           $\cdots$ & $\cdots$                           &  \\
\texttt{+440..000} & $+\frac{8}{15} \texttt{FF0..000}$ &  \\
\texttt{+440..000} &  &  \\
\texttt{+400..000} & $+\frac{8}{15} \texttt{F00..000}$ & \\
\texttt{+400..000} & &
\end{tabular}
}
\end{center}
Therefore,
\begin{align*}
    X_n &= \frac{2}{15}(16^n-1) - \frac{8}{15^2}(16^n-16) + \frac{8}{15}(n-1)16^n\\
        &= \frac{2}{225}(16^n(60n-49)+49).
\end{align*}
Similarly the remaining helper functions can be reduced to:
\setlength\extrarowheight{0pt}
\begin{align*}
Y_n &= \frac{2}{225}(16^n(60n-19)+19) \\
Z_n &= \frac{4}{225}(16^n(30n-17)+17) \\
W_n &= \frac{8}{225}(16^n(15n-1)+1)
\end{align*}

\begin{lemma}\label{lem:intersect_abstr}
Extremal values of $\Delta$ occur when $a_{i}\neq a_{i-2}$ and $b_{j}\neq b_{j-2}$ for all $n>i,j>1$.
\end{lemma}
\begin{proof}
Consider the contribution of $a_0$ and $a_2$ in $\Delta$:
\begin{align*}
\Delta = & (-2^{k-1}b_{k-1}+b_{k-2:0})a_0 \\
&+4(-2^{k-3}b_{k-3}+b_{k-4:0})a_2+... \\
= & -2^{k-1}(b_{k-1}a_0+b_{k-3}a_2) \\
& +2^{k-2}(b_{k-2}a_0+b_{k-4}a_2) + ... \\
& +4(b_2 a_0 + b_0 a_2) + ...
\end{align*}
For $a_0=a_2=1$, extremal $\Delta$ occurs when $b_{k-1:2}=b_{k-3:0}$. If $a_0=a_2=0$ made $\Delta$ extremal this would imply $b_{k-1:2}=b_{k-3:0}$. But if $a_0=a_2$ implies $b_{k-1:2}=b_{k-3:0}$ and hence $b_2=b_0$ then by a symmetric argument this would imply $a_{k-1:2}=a_{k-3:0}$. Up to input reordering, this implies the following extremal inputs:

\begin{table}[!h]
\begin{center}
\caption{$\Delta$ WORST CASE INPUTS $a_i=a_{i-2}$ AND $b_j=b_{j-2}$}\label{tab:worstdelateq}
\setlength\extrarowheight{4pt}
\scalebox{0.875}{
\begin{tabular}{|r|r|c|}
\hline
$a$ & $b$ & $\Delta$ \\
\hline
\texttt{}{111...111} & \texttt{111...111} & $1$ \\
\texttt{111...111} & \texttt{101...010} & $\frac{2^k+2}{3}$ \\
\texttt{111...111} & \texttt{010...101} & $-\frac{2^k-1}{3}$ \\
\texttt{101...010} & \texttt{101...010} & $\frac{1}{18}((3k+16)2^k+16)$ \\
\texttt{101...010} & \texttt{010...101} & $-k2^{k-2}-\frac{1}{3}(2^k-1)$ \\
\texttt{010...101} & \texttt{010...101} & $\frac{1}{18}((3k-2)2^k+2)$ \\
\texttt{000...000} & \texttt{000...000} & $0$ \\
\hline
\end{tabular}
}
\end{center}
\end{table}

Alternatively if $a_0 \neq a_2$, swapping the value of $a_0$ with $a_2$ will produce the largest change to $\Delta$ if $b_{k-1:2} \neq b_{k-3:0}$. But if $a_0 \neq a_2$ implies $b_{k-1:2} \neq b_{k-3:0}$ and hence $b_2 \neq b_0$ then by a symmetric argument this would imply $a_{k-1:2} \neq a_{k-3:0}$.

As an example, consider $k/2$ is odd and $a=b=10011001...100110$:

\begin{align*}
\Delta &= \sum_{i,j=0, i+j < k/2}^{(k/2)-1} 4^{i+j} PP_{i,j} \\
&= -\sum_{i,j=0, i+j < k/2}^{(k/2)-1} (-4)^{i+j+1} \\
&= Z_{(k+2)/4} - 4W_{(k-2)/4} \\
&= \frac{2}{25} (2^k(5k+2)+2)
\end{align*}
Similarly, the helper functions can be used to simplify $\Delta$ for all the cases where $a_{i}\neq a_{i-2}$ and $b_{j}\neq b_{j-2}$ for all $n>i,j>1$, four of these cases are:

\begin{table}[!h]
\begin{center}
\caption{$\Delta$ WORST CASE INPUTS $a_i \neq a_{i-2}$ AND $b_j \neq b_{j-2}$}\label{tab:worstdelatneq}
\setlength\extrarowheight{4pt}
\scalebox{0.875}{
\begin{tabular}{|c|l|l|c|}
\hline
$k/2$ & $a$ & $b$ & $\Delta$ \\
\hline
even & \texttt{10011001...1001} & \texttt{01100110...0110} & $4Y_{\frac{k}{4}}-X_{\frac{k}{4}}$ \\
odd  & \texttt{10011001...10}   & \texttt{10011001...10}   & $Z_{\frac{k+2}{4}}-4W_{\frac{k-2}{4}}$ \\
even & \texttt{01100110...0110} & \texttt{01100110...0110} & $Z_{\frac{k}{4}}-4W_{\frac{k}{4}}$ \\
odd  & \texttt{10011001...10}   & \texttt{01100110...01}   & $-X_{\frac{k+2}{4}}+4Y_{\frac{k-2}{4}}$ \\
\hline
\end{tabular}
}
\end{center}
\end{table}
The remaining cases produce less extremal $\Delta$ values. These four cases can be combined and simplified using the helper function to conclude:

\begin{align} \label{eq:deltarange}
-\frac{1}{25} (2^k(10k+5(-1)^{k/2}&-1)-5+(-1)^{k/2} \nonumber \\
\leq \Delta & \leq \nonumber \\
\frac{1}{25} (2^k(10k-5(-1)^{k/2}&-1)+5+(-1)^{k/2}
\end{align}

Note that the leading coefficient of $k2^k$ in these bounds are $2/5$ versus $1/6$ and $1/4$ in Table~\ref{tab:worstdelateq} . Conclude that the extremal values of $\Delta$ occur when $a_{i}\neq a_{i-2}$ and $b_{j}\neq b_{j-2}$ for all $n>i,j>1$. The worst case input vectors can be found in Table~\ref{tab:worstdelatneq} and \eqref{eq:deltarange} contains tight bounds on $\Delta$.
\end{proof}

\subsection{Faithfully Rounded Commutative Truncated Booth Arrays}

The optimal truncation scheme will truncate the most number of columns while maintaining the faithful rounding condition. The truncated array $M$ will have an additional constant $C$ added to compensate for the loss of $\Delta$. Once summed, the least significant bits (value $D$) will be truncated before returning the approximation.
\begin{align*}
&y = a \times b = M + \Delta \quad \text{Unrounded Result}& \\
&y' = M + C - D  \quad \text{Rounded Result}& \\
& |y - y'| < 2^n \quad \text{Faithful Rounding Condition}& \\
&\Rightarrow C-D-2^n < \Delta < C-D+2^n
\end{align*}
$D$ can take any value between $0$ and $2^n-2^k$, hence necessary and sufficient condition for faithful rounding is:
\begin{align*}
C-2^n < &\Delta < C+2^k \\
\max(\Delta)-2^k < &C < \min(\Delta) + 2^n
\end{align*}
Now $C$ is a multiple of $2^k$, hence:
\begin{align*}
\left \lfloor \frac{\max(\Delta)}{2^k} \right \rfloor \leq  \left \lfloor \frac{\min(\Delta)}{2^k} \right \rfloor + 2^{n-k}
\end{align*}
Substituting the extremal found in \eqref{eq:deltarange} and simplifying results in:
\begin{align} \label{eq:krange}
\left \lfloor \frac{2k-(-1)^{k/2}}{5} \right \rfloor + \left \lceil \frac{2k+(-1)^{k/2}}{5} \right \rceil < 2^{n-k}
\end{align}
Maximising $k$ in \eqref{eq:krange} generates the optimal truncation value. The optimal truncation scheme values can now be presented:
\begin{align}
& k^* = \max_{\text{even } k} \left(k\leq5 \times 2^{n-k-2} \right) \\
& \frac{2k^*-5-(-1)^{\frac{k^*}{2}}}{5} \leq C^* \leq \frac{5 \times 2^{n-k^*}-2k^*-(-1)^{\frac{k^*}{2}}}{5}
\end{align}
Example values of $k^*$ and $C^*$:
\begin{center}
\setlength\extrarowheight{4pt}
\scalebox{0.875}{
\begin{tabular}{|r|r|r|r|}
\hline
$n$ & $k^*$ & $\min C^*$ & $\max C^*$ \\
\hline
8 & 4 & 1 & 14 \\
16 & 12 & 4 & 11 \\
24 & 20 & 7 & 7 \\
32 & 26 & 10 & 53 \\
53 & 46 & 18 & 109 \\
64 & 58 & 23 & 41 \\
\hline
\end{tabular}
}
\end{center}

We can now present the steps in constructing a commutative truncated Booth Radix-4 array for an $n$ multiplication returning a faithful rounding of the $n$ most significant bits:
\begin{enumerate}
    \item Construct standard Booth Radix-4 array, Booth encoding $a$
    \item Calculate $k^* = \max_{\text{even } k} (k\leq5 \times 2^{n-k-2} )$
    \item Remove least significant $k^*$ columns
    \item  In column $k^*$, for $i \in [0,\frac{k^*}{2}-1]$, add additional bits 
\begin{equation*}
    s_i' = a_{2i+1} \& \overline{a_{2i} \& a_{2i-1}} \& (b_{k^*-2i-1} \oplus a_{2i+1}).
\end{equation*}
    \item In column $k^*$, include constant $C^*$ which can be any value in the range:
\begin{equation*}
    \left[ \frac{2k^*-5-(-1)^{\frac{k^*}{2}}}{5} , \frac{5 \times 2^{n-k^*}-2k^*-(-1)^{\frac{k^*}{2}}}{5} \right]
\end{equation*}
    \item Sum the array
    \item Remove the least significant $n$ columns
\end{enumerate}


\input{results_table}

\section{Results}\label{sec:results}

We compare three implementations of faithfully rounded commutative $n$-bit multipliers returning an $n$-bit result. The baseline is a round to zero multiplier, implemented by computing the full precision $n$-bit multiplication, generating a 2$n$-bit result, from which we return the $n$ most significant bits. The second implementation is a truncated AND array, where we follow the approach in~\cite{Drane2014OnArrays} to compute the maximum possible truncation and an efficient constant compensation term. We do not compare against alternative truncated Booth implementations since these designs do not retain commutativity.

\subsection{Synthesis}\label{subsec:synth}
We synthesize each design using a commercial logic synthesis tool, targeting a standard TSMC 5nm library. We present results for a range of bitwidths and at a number of delay targets for each parameterization. The combinational area and total power consumption results reported by the logic synthesis tool are shown in Table~\ref{tab:synth_table}.
We also present the percentage improvement in each metric with respect to the baseline. We synthesize multipliers at relevant bitwidths between 16 and 64 bits. At lower bitwidths the overhead of Booth encoding is detrimental~\cite{Zimmermann2009DatapathDesign} and at higher bitwidths multiplier decomposition techniques, such as Karatsuba~\cite{Karatsuba1962MultiplicationComputers}, dominate~\cite{Ustun2022IMpress:HLS}. 

In Table~\ref{tab:synth_table} we can see that the truncated Booth array is up to \maxareaboothbase~smaller than the baseline implementation and consumes up to \maxpowerboothbase~less power. Furthermore, the truncated Booth array is up to \maxareaboothand~smaller than the truncated AND array and consumes up to \maxpowerboothand~less power. For 16-bit multiplication, we see that the truncated AND array is both smaller and more power efficient across all delay targets. For 24-bit multiplication, the truncated Booth multiplier is superior. Prior work on exact multiplier implementations also observed an architectural crossover point around 16-bits~\cite{Zimmermann2009DatapathDesign}, at which the overhead of Booth encoding is offset by the gains in array reduction. As we increase the bitwidth, the benefit offered by the truncated Booth array over the alternatives increases.   

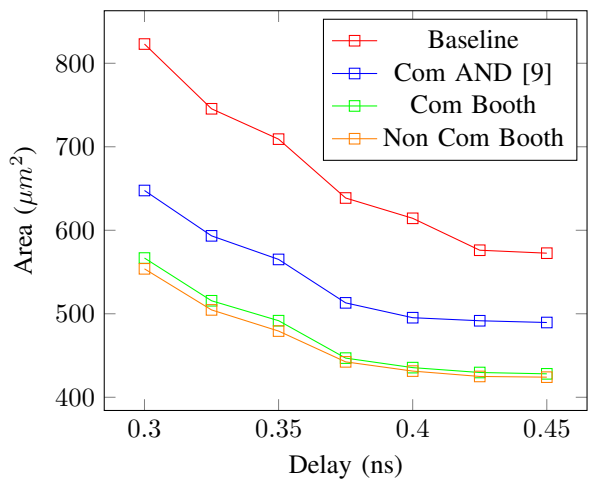
\begin{figure}
    \centering
    \input{area_delay_64bit}
    \caption{Area-delay profiles for the three competing commutative designs for n=64. We also plot the truncated Booth architecture without the compensation bits to recover commutativity.}
    \label{fig:area_delay_64}
\end{figure}

In Figure~\ref{fig:area_delay_64} we present complete area-delay profiles for the competing 64-bit multiplier implementations. Across the complete delay spectrum the truncated Booth multiplier demonstrates roughly constant area reduction when compared against the truncated AND array. 

To understand the penalty we must incur to recover commutativity, we also synthesize the 64-bit truncated Booth array without the additional compensation bits described in Section~\ref{subsec:compensation}. As we can see in Figure~\ref{fig:area_delay_64}, the area difference between the commutative and non-commutative truncated Booth implementations is minimal across the delay spectrum. At worst we pay a 2.5\% area penalty.

\subsection{Formal Verification}\label{subsec:verif}
For multiplier implementations beyond 16-bit, formal verification becomes challenging as the number of inputs exceeds what can be simulated. Furthermore, the verification of custom multiplier implementations is particularly challenging due to the circuit complexity, leading to bespoke tools~\cite{Koelbl2009SolverChecking,Burch1991UsingMultipliers} and methods~\cite{Kaivola2002FormalMultiplier,Kaufmann2019VerifyingAlgebra}. 

In this work, we deploy the S-C-Rewriting method~\cite{Temel2020AutomatedMultipliers,Temel2021SoundMultipliers} built on the ACL2 theorem prover~\cite{Kaufmann1996ACL2:Nqthm}, supported by the Glucose SAT solver. We use the tool to prove that the output, \texttt{out}, of the Verilog code implementing a truncated Booth multiplier satisfies the following lemma:
\begin{lstlisting}
mult = a[n-1:0]*b[n-1:0]
lsbs = mult[  n-1:0]
msbs = mult[2*n-1:n]
 
lemma (lsbs==0) ? out==msbs
                : 0<=out-msbs<=1
\end{lstlisting}
Where  \texttt{a} and  \texttt{b} are $n$-bit, and  \texttt{mult} is $2n$-bit wide unsigned bit-vectors representing integer values. This lemma guarantees that \texttt{out} is a faithful rounding. We also use the same tool to prove commutativity of the truncated Booth multiplier.

\input{proof_runtime}
We prove the two properties for a range of bitwidths $n$ and present the proof runtimes in Table~\ref{tab:verification}. Unfortunately, as the bitwidth $n$ increases, the proof of faithful rounding runtimes grow exponentially, meaning we are unable to prove the correctness of the 64-bit truncated Booth multiplier.

\section{Conclusion}
This paper provides the first implementation of a commutative truncated Booth multiplier, that produces a faithfully rounded result. We first described how, for minimal circuit area overhead, we can recover commutativity, via the introduction of a small number of compensation bits. We then proved exact bounds on the maximal error due to Booth array truncation and used these bounds to calculate the maximum number of columns which can be truncated. Lastly, we described how the addition of a constant can compensate for the error induced by truncation. We synthesized a number of faithfully rounded multiplier implementations and were able to reduce circuit area by up to \maxareaboothand~and reduce power consumption by up to~\maxpowerboothand~when compared to the state of the art. Using an ACL2 based verification tool, we were able to prove the correctness of the commutative truncated Booth multipliers up to 42 bits. 

Future work will look to generalize the results here to arbitrary Booth encoding radices, e.g. Booth Radix-8. A further generalization will consider different error thresholds, as opposed to the faithful rounding considered here. We will also address the limitations of our verification, exploring proof decomposition techniques. Lastly, an approach in~\cite{Drane2014OnArrays} can be used to incorporate these multiplier components into larger hardware designs, for example a floating-point multiplier.

\bibliographystyle{IEEEtran}
\bibliography{references}

\end{document}

%% file: trunc_pp_array.tex
\newcommand\trun[1][]{\tikz\node[circle,draw,inner sep=0.2em]{};}
\newcommand\ppar[1][]{\tikz\node[circle,draw,fill,inner sep=0.2em]{};}
\newcommand\resu[1][]{\tikz\node[circle,draw,fill=blue,inner sep=0.2em]{};}
\newcommand\disc[1][]{\tikz\node[circle,draw,fill=cyan,inner sep=0.2em]{};}
\newcommand\one[1][]{\tikz\node[red,circle,inner sep=0.2em]{1};}
\setlength\tabcolsep{1pt}
\renewcommand{\arraystretch}{0.5}
\begin{tabular}{cccc cccc cccc cccc | cccc cccc cccc cccc}
\multicolumn{16}{c}{$\overbrace{\hspace{12.4em}}^{16}$}&\multicolumn{4}{c}{$\overbrace{\hspace{2.5em}}^4$}&\multicolumn{12}{c}{$\overbrace{\hspace{9.4em}}^{12}$}\\
&&&&&&&&&&&&&&&&\ppar&\ppar&\ppar&\ppar&\trun&\trun&\trun&\trun&\trun&\trun&\trun&\trun&\trun&\trun&\trun&\trun\\ 
&&&&&&&&&&&&&&&\ppar&\ppar&\ppar&\ppar&\ppar&\trun&\trun&\trun&\trun&\trun&\trun&\trun&\trun&\trun&\trun&\trun&\\ 
&&&&&&&&&&&&&&\ppar&\ppar&\ppar&\ppar&\ppar&\ppar&\trun&\trun&\trun&\trun&\trun&\trun&\trun&\trun&\trun&\trun&&\\ 
&&&&&&&&&&&&&\ppar&\ppar&\ppar&\ppar&\ppar&\ppar&\ppar&\trun&\trun&\trun&\trun&\trun&\trun&\trun&\trun&\trun&&&\\ 

&&&&&&&&&&&&\ppar&\ppar&\ppar&\ppar&\ppar&\ppar&\ppar&\ppar&\trun&\trun&\trun&\trun&\trun&\trun&\trun&\trun&&&&\\ 
&&&&&&&&&&&\ppar&\ppar&\ppar&\ppar&\ppar&\ppar&\ppar&\ppar&\ppar&\trun&\trun&\trun&\trun&\trun&\trun&\trun&&&&&\\ 
&&&&&&&&&&\ppar&\ppar&\ppar&\ppar&\ppar&\ppar&\ppar&\ppar&\ppar&\ppar&\trun&\trun&\trun&\trun&\trun&\trun&&&&&&\\ 
&&&&&&&&&\ppar&\ppar&\ppar&\ppar&\ppar&\ppar&\ppar&\ppar&\ppar&\ppar&\ppar&\trun&\trun&\trun&\trun&\trun&&&&&&&\\ 

&&&&&&&&\ppar&\ppar&\ppar&\ppar&\ppar&\ppar&\ppar&\ppar&\ppar&\ppar&\ppar&\ppar&\trun&\trun&\trun&\trun&&&&&&&&\\ 
&&&&&&&\ppar&\ppar&\ppar&\ppar&\ppar&\ppar&\ppar&\ppar&\ppar&\ppar&\ppar&\ppar&\ppar&\trun&\trun&\trun&&&&&&&&&\\ 
&&&&&&\ppar&\ppar&\ppar&\ppar&\ppar&\ppar&\ppar&\ppar&\ppar&\ppar&\ppar&\ppar&\ppar&\ppar&\trun&\trun&&&&&&&&&&\\ 
&&&&&\ppar&\ppar&\ppar&\ppar&\ppar&\ppar&\ppar&\ppar&\ppar&\ppar&\ppar&\ppar&\ppar&\ppar&\ppar&\trun&&&&&&&&&&&\\ 

&&&&\ppar&\ppar&\ppar&\ppar&\ppar&\ppar&\ppar&\ppar&\ppar&\ppar&\ppar&\ppar&\ppar&\ppar&\ppar&\ppar&&&&&&&&&&&&\\ 
&&&\ppar&\ppar&\ppar&\ppar&\ppar&\ppar&\ppar&\ppar&\ppar&\ppar&\ppar&\ppar&\ppar&\ppar&\ppar&\ppar&\color{red} 1&&&&&&&&&&&&\\ 
&&\ppar&\ppar&\ppar&\ppar&\ppar&\ppar&\ppar&\ppar&\ppar&\ppar&\ppar&\ppar&\ppar&\ppar&\ppar&\ppar&\color{red} 1&\color{red} 1&&&&&&&&&&&&\\ 
&\ppar&\ppar&\ppar&\ppar&\ppar&\ppar&\ppar&\ppar&\ppar&\ppar&\ppar&\ppar&\ppar&\ppar&\ppar&\ppar&\color{red} 1&\color{red} 1&\color{red} 1&&&&&&&&&&&&\\ 
\hline
\resu&\resu&\resu&\resu&\resu&\resu&\resu&\resu&\resu&\resu&\resu&\resu&\resu&\resu&\resu&\resu&\disc&\disc&\disc&\disc\\
\end{tabular}

%% file: compensated_booth.tex
\newcommand\trun[1][]{\tikz\node[circle,draw,inner sep=0.2em]{};}
\newcommand\ppar[1][]{\tikz\node[circle,draw,fill,inner sep=0.2em]{};}
\newcommand\comp[1][]{\tikz\node[circle,draw,fill=red,inner sep=0.2em]{};}
\newcommand\corr[1][]{\tikz\node[circle,draw,fill,inner sep=0.2em]{};}
\newcommand\sign[1][]{\tikz\node[circle,draw,fill=magenta!50,inner sep=0.2em]{};}
\newcommand\resu[1][]{\tikz\node[circle,draw,fill=blue,inner sep=0.2em]{};}
\newcommand\disc[1][]{\tikz\node[circle,draw,fill=cyan,inner sep=0.2em]{};}
\newcommand\one[1][]{\tikz\node[red,circle,inner sep=0.2em]{1};}
\setlength\tabcolsep{1pt}
\renewcommand{\arraystretch}{0.5}
\begin{tabular}{cccc cccc cccc cccc | cccc cccc cccc cccc}
\multicolumn{16}{c}{$\overbrace{\hspace{12.4em}}^{16}$}&\multicolumn{4}{c}{$\overbrace{\hspace{2.5em}}^4$}&\multicolumn{12}{c}{$\overbrace{\hspace{9.4em}}^{12}$}\\
&&&&&&&&&&&&&&\sign&\sign&\ppar&\ppar&\ppar&\corr&\trun&\trun&\trun&\trun&\trun&\trun&\trun&\trun&\trun&\trun&\trun&\trun\\ 
&&&&&&&&&&&&&&&&&&&\comp&&&&&&&&&&&&\trun\\ 
&&&&&&&&&&&&\sign&\sign&\ppar&\ppar&\ppar&\ppar&\ppar&\corr&\trun&\trun&\trun&\trun&\trun&\trun&\trun&\trun&\trun&\trun&&\\ 
&&&&&&&&&&&&&&&&&&&\comp&&&&&&&&&&\trun&&\\ 

&&&&&&&&&&\sign&\sign&\ppar&\ppar&\ppar&\ppar&\ppar&\ppar&\ppar&\corr&\trun&\trun&\trun&\trun&\trun&\trun&\trun&\trun&&&&\\ 
&&&&&&&&&&&&&&&&&&&\comp&&&&&&&&\trun&&&&\\ 
&&&&&&&&\sign&\sign&\ppar&\ppar&\ppar&\ppar&\ppar&\ppar&\ppar&\ppar&\ppar&\corr&\trun&\trun&\trun&\trun&\trun&\trun&&&&&&\\ 
&&&&&&&&&&&&&&&&&&&\comp&&&&&&\trun&&&&&&\\ 

&&&&&&\sign&\sign&\ppar&\ppar&\ppar&\ppar&\ppar&\ppar&\ppar&\ppar&\ppar&\ppar&\ppar&\corr&\trun&\trun&\trun&\trun&&&&&&&&\\ 
&&&&&&&&&&&&&&&&&&&\comp&&&&\trun&&&&&&&&\\ 
&&&&\sign&\sign&\ppar&\ppar&\ppar&\ppar&\ppar&\ppar&\ppar&\ppar&\ppar&\ppar&\ppar&\ppar&\ppar&\corr&\trun&\trun&&&&&&&&&&\\ 
&&&&&&&&&&&&&&&&&&&\comp&&\trun&&&&&&&&&&\\ 

&&\sign&\sign&\ppar&\ppar&\ppar&\ppar&\ppar&\ppar&\ppar&\ppar&\ppar&\ppar&\ppar&\ppar&\ppar&\ppar&\ppar&\ppar&&&&&&&&&&&&\\ 
&&&&&&&&&&&&&&&&&&&\sign&&&&&&&&&&&&\\ 
&&\ppar&\ppar&\ppar&\ppar&\ppar&\ppar&\ppar&\ppar&\ppar&\ppar&\ppar&\ppar&\ppar&\ppar&\ppar&\ppar&&&&&&&&&&&&&&\\ 
&&&&&&&&&&&&&&&&&\sign&\multicolumn{2}{c}{\color{red} \footnotesize$C^*$}&&&&&&&&&&&\\ 
\hline
\resu&\resu&\resu&\resu&\resu&\resu&\resu&\resu&\resu&\resu&\resu&\resu&\resu&\resu&\resu&\resu&\disc&\disc&\disc&\disc\\
\end{tabular}

%% file: results_table.tex
\begin{table*}
    \centering
    \caption{Synthesis results for three competing architectures at several delay targets, across four different bitwidths, $n$. The percentage improvements are with respect to the baseline implementation. We highlight the best result in each row in bold.}
    
    \begin{tabular}{|cc|rr|rr|rr|}
    \hline
         & & \multicolumn{2}{c|}{Baseline} & \multicolumn{2}{c|}{Truncated AND~\cite{Drane2014OnArrays}} & \multicolumn{2}{c|}{Commutative Truncated Booth} \\
         $n$ & Delay (ns) & Area ($\mu m^2$) & Power ($\mu W$) & Area ($\mu m^2$) & Power ($\mu W$) & Area ($\mu m^2$) & Power ($\mu W$) \\   
         \hline
         \multirow{4}{*}{16}
         &0.175& 74.0 & 450 & \best{64.4} (\textit{-13.0\%}) & \best{319} (\textit{-29.2\%}) & 69.3 (\textit{- 6.4\%}) & 389 (\textit{-13.7\%})\\
         &0.2  & 56.1 & 346 & \best{46.4} (\textit{-17.4\%}) & \best{252} (\textit{-27.3\%}) & 48.3 (\textit{-13.9\%}) & 286 (\textit{-17.5\%})\\
         &0.225& 51.7 & 309 & \best{40.9} (\textit{-20.8\%}) & \best{215} (\textit{-30.4\%}) & 43.6 (\textit{-15.8\%}) & 252 (\textit{-18.5\%})\\
         &0.25 & 47.5 & 283 & \best{37.1} (\textit{-21.9\%}) & \best{198} (\textit{-30.1\%}) & 39.7 (\textit{-16.3\%}) & 226 (\textit{-20.1\%})\\

         \hline
         \multirow{4}{*}{24} 
         &0.225& 127.7 & 815 & 108.6 (\textit{-14.9\%}) & 622 (\textit{-23.7\%}) & \best{99.9} (\textit{-21.8\%}) & \best{600} (\textit{-26.4\%})\\
         &0.25 & 113.4 & 705 & 101.5 (\textit{-10.4\%}) & 575 (\textit{-18.4\%}) & \best{88.7} (\textit{-21.8\%}) & \best{516} (\textit{-26.8\%})\\
         &0.275& 108.0 & 653 & 93.5  (\textit{-13.4\%}) & 502 (\textit{-23.2\%}) & \best{84.9} (\textit{-21.3\%}) & \best{498} (\textit{-23.8\%})\\
         &0.3  & 99.2  & 595 & 83.8  (\textit{-15.6\%}) & 466 (\textit{-21.8\%}) & \best{77.3} (\textit{-22.1\%}) & \best{447} (\textit{-25.0\%})\\

         \hline
         \multirow{4}{*}{32} 
         &0.275& 193.7 & 1225 & 161.3 (\textit{-16.7\%}) & 897 (\textit{-26.7\%}) & \best{156.2} (\textit{-19.4\%}) & \best{914} (\textit{-25.4\%})\\
         &0.3  & 171.5 & 1097 & 156.3 (\textit{- 8.8\%}) & 861 (\textit{-21.5\%}) & \best{148.0} (\textit{-13.7\%}) & \best{862} (\textit{-21.4\%})\\
         &0.325& 164.3 & 1024 & 139.7 (\textit{-15.0\%}) & 787 (\textit{-23.1\%}) & \best{134.2} (\textit{-18.3\%}) & \best{768} (\textit{-25.0\%})\\
         &0.35 & 152.9 & 945  & 137.2 (\textit{-10.3\%}) & 772 (\textit{-18.3\%}) & \best{130.7} (\textit{-14.5\%}) & \best{748} (\textit{-20.8\%})\\

         \hline
         \multirow{4}{*}{64} 
         &0.3  & 823.0 & 5548 & 647.6 (\textit{-21.3\%}) & 3808 (\textit{-31.4\%}) & \best{566.8} (\textit{-31.1\%}) & \best{3527} (\textit{-36.4\%})\\
         &0.325& 745.5 & 4886 & 593.3 (\textit{-20.4\%}) & 3403 (\textit{-30.4\%}) & \best{515.6} (\textit{-30.8\%}) & \best{3073} (\textit{-37.1\%})\\
         &0.35 & 709.2 & 4558 & 565.1 (\textit{-20.3\%}) & 3165 (\textit{-30.6\%}) & \best{491.7} (\textit{-30.7\%}) & \best{2846} (\textit{-37.6\%})\\
         &0.375& 638.6 & 4187 & 513.1 (\textit{-19.7\%}) & 2936 (\textit{-29.9\%}) & \best{446.9} (\textit{-30.0\%}) & \best{2593} (\textit{-38.1\%})\\

         \hline
    \end{tabular}
    \label{tab:synth_table}
\end{table*}

%% file: area_delay_64bit.tex
\begin{tikzpicture}[]
	\begin{axis}[%
		xlabel=Delay (ns),ylabel=Area ($\mu m^2$)]

\addplot[
    color=red,
    mark=square,
    ]
    coordinates {
(0.3,	823.031371)
(0.325,	745.501684)
(0.35,	709.151945)
(0.375,	638.648012)
(0.4,	614.325602)
(0.425,	576.10161)
(0.45,	572.56731)

};
\addlegendentry{Baseline}

\addplot[
    color=blue,
    mark=square,
    ]
    coordinates {
(0.3,	647.633701)
(0.325,	593.334004)
(0.35,	565.081025)
(0.375,	513.06255)
(0.4,	495.29466)
(0.425,	491.63184)
(0.45,	489.54339)

    };
    \addlegendentry{Com AND~\cite{Drane2014OnArrays}}

\addplot[
    color=green,
    mark=square,
    ]
    coordinates {
(0.3,	566.751777)
(0.325,	515.579401)
(0.35,	491.717522)
(0.375,	446.90688)
(0.4,	435.51144)
(0.425,	429.67449)
(0.45,	428.06799)
    };
    \addlegendentry{Com Booth}

\addplot[
    color=orange,
    mark=square,
    ]
    coordinates {
(0.3,553.814097)
(0.325,504.623071)
(0.35,479.283212)
(0.375,442.39797)
(0.4,431.35596)
(0.425,424.92996)
(0.45,424.08387)
    };
    \addlegendentry{Non Com Booth}

    \end{axis}
\end{tikzpicture}

%% file: proof_runtime.tex
\begin{table}
    \centering
    \caption{ACL2 runtimes for proving that the truncated booth multiplier implements a faithful rounding and is commutative. The dash indicates a proof which did not converge.}
    \begin{tabular}{r|r|r}
         $n$ & Faithful (sec) & Commutative (sec)  \\
         \hline
            4	&3	 & 0.6 \\
            6	&4	 & 0.6 \\
            16	&7	 & 0.7 \\
            24	&44	 & 0.7 \\
            32	&63	 & 0.9 \\
            36	&117 & 1.0 \\
            42	&835 & 1.1\\
            64	&--	 & 2.0\\
    \end{tabular}
    
    \label{tab:verification}
\end{table}